\UseRawInputEncoding
\documentclass[11pt, reqno]{amsart}

\usepackage{amsfonts,latexsym, enumerate}
\usepackage{amsmath}
\usepackage{amscd}
\usepackage{float,amsmath,amssymb,mathrsfs,bm,multirow,graphics}
\usepackage[dvips]{graphicx}
\usepackage[percent]{overpic}
\usepackage{amsaddr}
\usepackage{todonotes}
\usepackage[numbers,sort&compress]{natbib}
\usepackage{grffile}

\newcommand{\R}{{\mathbb R}}

\newcommand{\C}{{\mathbb C}}
\newcommand{\Z}{{\mathbb Z}}

\DeclareMathOperator{\arctanh}{arctanh}

\newcommand{\Li}{\text{\upshape Li}}

\newtheorem{theorem}{Theorem}[section]
\newtheorem{lemma}[theorem]{Lemma}
\newtheorem{proposition}[theorem]{Proposition}

\theoremstyle{definition}

\theoremstyle{remark}

\numberwithin{equation}{section}

\usepackage[colorlinks=true]{hyperref}
\hypersetup{urlcolor=blue, citecolor=red, linkcolor=blue}

\input epsf
\title[Antiferromagnetic gap for half-filled 2D Hubbard model]{On the mean-field antiferromagnetic gap for the half-filled 2D Hubbard model at zero temperature}

\author{E. Langmann$^{1}$ and J. Lenells$^{2}$}
\address{$^1$Department of Physics, KTH Royal Institute of Technology, \\ 106 91 Stockholm, Sweden
	\\
$^{2}$Department of Mathematics, KTH Royal Institute of Technology, \\ 100 44 Stockholm, Sweden}
\email{langmann@kth.se}
\email{jlenells@kth.se}
    
\begin{document}

\begin{abstract}
We consider the antiferromagnetic gap for the half-filled two-dimensional (2D) Hubbard model (on a square lattice) at zero temperature in Hartree--Fock theory. It was conjectured by Hirsch in 1985 that this gap, $\Delta$, vanishes like $\exp(-2\pi\sqrt{t/U})$ in the weak-coupling limit $U/t\downarrow 0$ ($U>0$ and $t>0$ are the usual Hubbard model parameters). We give a proof of this conjecture based on recent mathematical results about Hartree-Fock theory for the 2D Hubbard model. The key step is the exact computation of an integral involving the density of states of the 2D tight binding band relation.
\end{abstract}

\maketitle


 
\section{Introduction}
The 2D Hubbard model is a basic model for magnetism and a central model in present theories of high-temperature superconductors, see e.g. \cite{ABKR2022, C2023, QSACG2022} for recent reviews. The 1D Hubbard model can been solved exactly by means of the Bethe Ansatz \cite{LW1968}, but in two and higher dimensions the model turns out to be remarkably rich and displays a plethora of interesting phenomena. For example, it is generally believed that the half-filled Hubbard model describes a Mott insulator at half-filling \cite{M1968}, as strongly suggested by Hartree-Fock theory; see \cite{BLS1994,BP1996,LL2024} and references therein.

In this note, we consider the half-filled 2D Hubbard model at zero temperature in the thermodynamic limit. In the repulsive case where the Hubbard coupling parameter,  $U$, is positive, the ground state of this model exhibits antiferromagnetic order, meaning that the expected spin density $\vec{m}(x_1, x_2)$ at the lattice site $(x_1, x_2) \in \Z^2$ has the form
\begin{align}\label{vecmdef}
\vec{m}(x_1, x_2) = m (-1)^{x_1 + x_2}  \vec{e},
\end{align}
where $m \in (0,1]$ is a constant and $\vec{e} \in \R^3$ is a fixed unit vector. 
In the Hartree--Fock approximation, the antiferromagnetic energy gap $\Delta = \frac{U m }{2}$ is determined by the well-known equation (see \cite{H1985} and references therein; a derivation can be found in \cite{LL2024}, e.g.)
\begin{align}\label{DeltaAFeq}
\frac{1}{U} =  \frac{1}{(2\pi)^2}\int_{-\pi}^\pi \int_{-\pi}^\pi 
\frac{1}{2\sqrt{\Delta^2 + \varepsilon(k_1, k_2)^2}}  dk_1 dk_2,
\end{align}
where $U > 0$ is the on-site repulsion and $\varepsilon(k_1, k_2)$ is the dispersion relation, which can be expressed in terms of the hopping parameter $t > 0$ as
\begin{align}\label{epsilondef}
\varepsilon(k_1, k_2) := - 2 t(\cos{k_1} + \cos{k_2}).
\end{align}
If $\Delta(U,t)$ is the antiferromagnetic energy gap associated with the parameters $(U,t)$, then $\Delta$  possesses the scaling property $\Delta(U,t) = t\Delta(U/t,1)$, implying that we may without of generality assume that $t = 1$ in what follows. We may then view $\Delta = \Delta(U)$ as a function of the coupling parameter $U>0$. 
It follows from (\ref{DeltaAFeq}) and relatively straightforward methods (see Section \ref{proofsec1}) that, up to exponentially small corrections,
\begin{align}\label{DeltaUb1}
\Delta(U) \approx 32 e^{-\sqrt{\frac{4 \pi^2}{U} + b_1}} \qquad \text{as $U \downarrow 0$},
\end{align}
where $b_1$ is a real constant given by
\begin{align}\label{b1def}
b_1 := 4 \pi^2 \left(\frac{(\ln{2})^2}{\pi^2} -\frac{1}{24} -a_0\right).
\end{align}
Here $a_0$ is a real constant defined in terms of the density of states 
\begin{equation}\label{N0}  
N_0(\epsilon) :=  \frac{1}{(2\pi)^2} \int_{-\pi}^\pi \int_{-\pi}^\pi \delta\big(2 (\cos{k_1} + \cos{k_2}) + \epsilon\big) dk_1dk_2
\end{equation} 
by
\begin{align}\label{a0def}
a_0 := \int_{0}^{4} \bigg(N_0(\epsilon) - \frac{\ln(\frac{16}{\epsilon})}{2 \pi^2}\bigg)
\frac{1}{\epsilon}  d\epsilon
\end{align}
($\delta$ in \eqref{N0} is the Dirac delta function).
To be mathematically precise, what we mean by (\ref{N0}) is that $N_0(\epsilon)d\epsilon$ is the pushforward to $\R$ by the function $(k_1,k_2) \mapsto -2 (\cos{k_1} + \cos{k_2})$ of the Lebesgue measure on $[-\pi,\pi]^2$ divided by $(2\pi)^2$.
This means that
\begin{align}\label{gN0}
\int_\R g(\epsilon) N_0(\epsilon)d\epsilon = \frac{1}{(2\pi)^2} \int_{-\pi}^\pi \int_{-\pi}^\pi g(-2 (\cos{k_1} + \cos{k_2})) dk_1dk_2
\end{align}
for any measurable function $g:\R \to \R$ such that $gN_0 \in L^1(\R)$.
We note that $N_0(\epsilon)$ is an interesting transcendental function \cite{DJ2001}, and we observe that the integral  defining $a_0$ converges near $\epsilon = 0$, because
\begin{align}\label{N0asymptotics}
N_0(\epsilon) 
= \frac{\ln(\frac{16}{\epsilon})}{2 \pi^2} + \frac{\epsilon^2 (\ln(\frac{16}{\epsilon})-1)}{128 \pi^2} + O\Big(\epsilon^{4} \ln{\tfrac{1}{\epsilon}}\Big) \qquad \text{as $\epsilon\downarrow 0$};
\end{align}
see \cite[Theorem 2.1]{LL2024}; the constant $a_0$ is a measure of what remains of the density of states after the logarithmic singularity at $\epsilon = 0$ has been subtracted.

In 1985, Hirsch presented a result without derivation suggesting that the constant $b_1$ defined in \eqref{b1def} equals zero, see \cite[Eq. (2.15)]{H1985}.
Furthermore, numerical computations indicate that $b_1$ vanishes to high accuracy. 
A negative value of $b_1$ would be highly surprising because it would indicate a breakdown of the asymptotic formula at $U = 4\pi^2/|b_1|$. 
This raises the question: is $b_1$ exactly zero or only approximately zero? 
And if it is exactly zero, where does this fine-tuning of the parameter $b_1$ come from?

In this note, we show that $b_1$ indeed is exactly zero. More precisely, we prove the following asymptotic formula for the antiferromagnetic energy gap which is our main result.

\begin{theorem}\label{mainth}
The antiferromagnetic energy gap $\Delta(U)$ of the 2D Hubbard model defined as the unique solution of (\ref{DeltaAFeq}) with $t = 1$ satisfies
\begin{align*}
	\Delta(U) 
	= 32 e^{-\frac{2\pi}{\sqrt{U}}}\Bigg(1 + O\bigg(\frac{e^{-\frac{4\pi}{\sqrt{U}}}}{\sqrt{U}}\bigg)\Bigg)
	\qquad \text{as $U\downarrow 0$.}
\end{align*}
\end{theorem}

The proof of Theorem \ref{mainth} will involve two steps. 
The first step is to prove the following proposition which provides a precise version of (\ref{DeltaUb1}).

\begin{proposition}\label{prop1}
The function $\Delta(U)$ defined by (\ref{DeltaAFeq}) with $t = 1$ satisfies
\begin{align*}
	\Delta(U) 
	= 32 e^{-\sqrt{\frac{4 \pi^2}{U} + b_1}}\Bigg(1 + O\bigg(\frac{e^{-\frac{4\pi}{\sqrt{U}}}}{\sqrt{U}}\bigg)\Bigg)
	\qquad \text{as $U\downarrow 0$},
\end{align*}
where $b_1$ is given by (\ref{b1def}).
\end{proposition}

The second step is to prove that the constant $b_1$ vanishes.

\begin{proposition}\label{a0prop}
The constant $b_1$ defined in (\ref{b1def}) equals $0$. 
\end{proposition}

Theorem \ref{mainth} is an immediate consequence of Propositions \ref{prop1} and \ref{a0prop}.

The proofs of Propositions \ref{prop1} and \ref{a0prop} are presented in Sections \ref{proofsec1} and \ref{proofsec2}, respectively.

While our results prove that $b_1=0$, the deeper reasons behind this fine-tuning of the parameter $b_1$ remain mysterious to us. We find it remarkable that $b_1$ is exactly zero, and we believe that there is an interesting explanation why this is true which remains to be found.

\subsection{Proof strategy}\label{strategy} 
While the proof of Proposition \ref{prop1} is relatively straightforward, it was an interesting challenge for us to prove Proposition \ref{a0prop}. Our evaluation is reminiscent of a quantum field theory calculation where a finite answer is extracted from an infinite integral by subtracting off the divergence. Indeed, the integral in (\ref{a0def}) converges because the divergent part $\ln(\frac{16}{\epsilon})/(2 \pi^2)$ has been subtracted from $N_0(\epsilon)$. To control the divergence, we use a regularization and, while many different regularizations exist which all would give the same final answer, the challenge was to find a regularization making the computations as simple as possible. 

In a first version of this paper, available on the arXiv,\footnote{\href{https://arxiv.org/abs/2501.18141v1}{https://arxiv.org/abs/2501.18141v1}} we proved Proposition~\ref{a0prop} by using a regularization that led us to challenging integrals we were able to analyze using asymptotic methods and complex analysis. Subsequently, we found another regularization leading to an exactly solvable integral that provides a much simpler proof, allowing us to reduce the size of this paper by a factor $2$. We believe that our first proof is still interesting since it can be generalized to situations where an exact computation of the integral is not available.

\subsection{Further applications in the 2D Hubbard model} 
The main result of this paper is the exact value of  the integral defined in \eqref{a0def} appearing in an asymptotic  formula for the mean-field gap of the 2D Hubbard model at half-filling; see \eqref{a0conjecture} . In this section, we point out other implications of this result which are of interest in physics. 

The N\'eel temperature, $T_N$,  of an antiferromagnetic system is defined as the transition temperature from the  antiferromagnetic to the paramagnetic (non-ordered) state. We recently obtained the following result for the mean-field N\'eel temperature of the 2D Hubbard model at half-filling (setting $t=1$), 
\begin{align}\label{TNexpansion2D}
	T_N(U)
	= \frac{32}{\pi e^{-\gamma}}
	e^{-\sqrt{ \frac{4 \pi^2}{U}  + a_1}}\Big(1 + O\Big(e^{-\frac{4\pi}{\sqrt{U}}}\Big)\Big) \qquad \text{as $U \downarrow 0$}, 
\end{align}
where $\gamma \approx 0.5772$ is Euler's gamma constant and the constant $a_1 \approx 0.3260$ is defined by
\begin{align} 
	a_1 :=  -4 \pi^2 a_0
	- \int_{0}^\infty \frac{(\ln x)^2}{\cosh^{2}{x}} dx + (2\ln2)^2 +(\gamma+2\ln2-\ln\pi)^2 \label{a1def} 
\end{align} 
with $a_0$ in \eqref{a0def}; see \cite[Theorem~2.4]{LL2024}. This result is similar to the one in Proposition~\ref{prop1}, and our definition in \eqref{b1def} implies 
\begin{equation}\label{a1}  
a_1 -b_1 = (\gamma+2\ln2-\ln\pi)^2 +  \frac{\pi^2}{6}- \int_{0}^\infty \frac{(\ln x)^2}{\cosh^{2}{x}} dx .
\end{equation} 
It is worth noting that this difference $a_1-b_1$ appears in an asymptotic formula for the gap ratio, $\Delta/T_N$; see \cite{LL2024} for a discussion of the significance of the quantity $\Delta/T_N$. Indeed, from Proposition~\ref{prop1} and \eqref{TNexpansion2D}, one can get the following mean-field gap ratio of the 2D Hubbard model by a simple computation,  
\begin{equation} 
\frac{\Delta(U)}{T_N(U)} = \pi e^{-\gamma} + \frac{e^{-\gamma}}{4}(a_1-b_1)\sqrt{U} + O(U) \qquad \text{as $U \downarrow 0$} \label{gapratiolong};  
\end{equation} 
since $b_1=0$ by Proposition~\ref{a0prop}, $a_1$ in \eqref{TNexpansion2D} is identical to $a_1-b_1$ in \eqref{a1}.

This paper is the second in a series, starting with \cite{LL2024}, where we develop analytic tools to study Hartree--Fock theory for Hubbard-like models. As a first step, we concentrate on half-filling and the standard Hubbard model since, in this case, it is known that the minimizer of the Hartree--Fock functional is translation invariant (by translations by two sites) \cite{BLS1994,BP1996}, which simplifies the analysis. 
In on-going work, we address other parts of the phase diagram where it is known that this minimizer can break translation invariance in complicated ways \cite{LL2025,CLL2025}; see \cite{LW2007} and references therein for previous work on this topic. 
	
It is interesting to note that the integral equation in (\ref{DeltaAFeq}) is the same as the equation determining the superconducting gap in the attractive 2D Hubbard model, $U<0$, obtained in the theory of Bardeen, Cooper and Schrieffer (BCS); see e.g.\ \cite{DHM2023, FHNS2007, HHSS2008,HS2008,LT2023,HL2025, HLR2024} for mathematical work on BCS theory. Thus, our results also have applications in BCS theory. 

\subsection{Notation}\label{notationsubsec}
The following notation will be used. 
We write $s \downarrow 0$ to denote the limit as $s$ approaches $0$ along the positive real axis. 
The principal branch is used for all complex powers, square roots, logarithms, and dilogarithms. In the case of the dilogarithm $\Li_2(z)$ (whose definition is reviewed in Appendix \ref{polylogapp}), this means that $\Li_2(z)$ has a branch cut along $[1,+\infty)$. 
We use $C > 0$ to denote a generic constant whose value may change within a computation. 

\section{Asymptotics of $\Delta$: proof of Proposition \ref{prop1}}\label{proofsec1}
Setting $t = 1$ and rewriting (\ref{DeltaAFeq}) in terms of the density of states $N_0(\epsilon)$, we see that $\Delta(U)$ is defined for $U > 0$ as the unique positive solution of the equation
\begin{align}\label{premAFeq}
\frac{1}{U} = \int_\R N_{0}(\epsilon)  
\frac{1}{2\sqrt{\Delta^2 + \epsilon^2}}  d\epsilon.
\end{align} 
Since $N_0(\epsilon)$ is an even function of $\epsilon$ supported on the interval $[-4,4]$, equation (\ref{premAFeq}) simplifies to
\begin{align}\label{mAFeq}
\frac{1}{U} = \int_0^4 N_{0}(\epsilon)  
\frac{1}{\sqrt{\Delta^2 + \epsilon^2}}  d\epsilon.
\end{align} 
To extract the small $U$ behavior, we write (\ref{mAFeq}) as
\begin{align}\label{a0I1I2}
\frac{1}{U} = &\; a_0 + I_1(\Delta) + I_2(\Delta) ,
\end{align}
where $a_0 \approx 0.007013$ is the constant in (\ref{a0def}) and
\begin{align*}
& I_1(\Delta) := \int_0^4 \frac{\ln(\frac{16}{\epsilon})}{2 \pi^2} 
\frac{1}{\sqrt{\Delta^2 + \epsilon^2}}  d\epsilon,
	\\
& I_2(\Delta) := \int_0^4 \bigg(N_{0}(\epsilon)  - 
\frac{\ln(\frac{16}{\epsilon})}{2 \pi^2}\bigg)
\bigg(\frac{1}{\sqrt{\Delta^2 + \epsilon^2}}
- \frac{1}{\epsilon}\bigg) d\epsilon.
\end{align*}
The integral $I_1(\Delta)$ can be computed exactly by noting that
\begin{align*}
& \frac{1}{16 \pi^2} \frac{d}{d\epsilon} \bigg\{
-2 \bigg(\ln \frac{1 + \sqrt{\frac{\Delta^2}{\epsilon^2}+1}}{2}\bigg)^2
-\bigg(\ln \frac{\Delta^2}{\epsilon^2}\bigg)^2
+4 \ln \bigg(\frac{1 + \sqrt{\frac{\Delta^2}{\epsilon^2}+1}}{2}\bigg) \ln \bigg(\frac{\Delta
^2}{\epsilon^2}\bigg)
	\\
& +4 \ln \bigg(\frac{256}{\Delta^2}\bigg) \arctanh\bigg(\frac{1}{\sqrt{\frac{\Delta^2}{\epsilon^2}+1}}\bigg)
+4  \Li_2\bigg(\frac{1 - \sqrt{\frac{\Delta^2}{\epsilon^2}+1}}{2}\bigg)\bigg\}
= \frac{\ln(\frac{16}{\epsilon})}{2 \pi^2} 
\frac{1}{\sqrt{\Delta^2 + \epsilon^2}}
\end{align*}
for $\epsilon > 0$ and $\Delta > 0$, where we have used (\ref{polylogderivatives}). Employing (\ref{Li2asymptotics}), this yields
\begin{align*}
I_1(\Delta) = &\; \frac{1}{24 \pi^2}\bigg\{12 \ln \left(\frac{16}{\Delta }\right) \arctanh\left(\frac{4}{\sqrt{\Delta^2+16}}\right)
	\\
& -3 \ln\left(\sqrt{\Delta^2+16}+4\right) \ln\left(\frac{4 \left(\sqrt{\Delta^2+16}+4\right)}{\Delta^4}\right)-6 \ln(\Delta ) \ln(4 \Delta )
	\\
& +6 \Li_2\left(\frac{1}{2}-\frac{\sqrt{\Delta^2+16}}{8}\right)+\pi^2+27 (\ln 2)^2\bigg\} \qquad \text{for $\Delta > 0$},
\end{align*}
and hence (since $\Li_2(z) = O(z)$ as $z \to 0$)
\begin{align}\label{I1estimate}
I_1(\Delta)
= \frac{6 (\ln \Delta )^{2} -60 \ln(2) \ln(\Delta )+\pi^2+126 (\ln 2)^{2}}{24 \pi^2} + O(\Delta^2)  \qquad \text{as $\Delta \downarrow 0$}.
\end{align}

The second integral $I_2(\Delta)$ can be estimated by using the bound $|N_0(\epsilon)-\frac{1}{2\pi^2}\ln(\frac{16}{\epsilon})|\leq C\epsilon^2\ln(\frac{16}{\epsilon})$ ($0 < \epsilon \leq 4$) implied by (\ref{N0asymptotics}), the change of variables $x=\epsilon/\Delta$, and the inequality 
$x^2(x^{-1}-(1+x^{2})^{-1/2})\leq \min(x,x^{-1})$ for $x > 0$. This gives, for $0<\Delta\leq 4$,
\begin{align}\nonumber
	|I_2(\Delta)|\leq& \int_0^{4} C\epsilon^2\ln(\tfrac{16}{\epsilon})\left(\frac1{\epsilon} -\frac1{\sqrt{\Delta^2+\epsilon^2}}\right)d\epsilon 
		\\ \nonumber
	= &\; C\Delta^2\int_{0}^{4/\Delta}
	x^2\left(\frac1x-\frac1{\sqrt{1+x^2}} \right) \ln(\tfrac{16}{x\Delta})dx
	\\ \label{I2estimate}
	 \leq 
	 &\; C\Delta^2\int_0^1  x\ln(\tfrac{16}{x\Delta})dx
	+ C\Delta^2\int_1^{\frac4\Delta}x^{-1}\ln(\tfrac{16}{x\Delta})dx
	\leq C \Delta^2 |\ln\Delta|^2.
\end{align}

Substitution of (\ref{I1estimate}) and (\ref{I2estimate}) into (\ref{a0I1I2}) yields
\begin{align*}
		\frac{1}{U} = a_0 +
		\frac{6 (\ln \Delta )^{2} -60 \ln(2) \ln(\Delta )+\pi^2+126 (\ln 2 )^{2}}{24 \pi^2}
		+ \frac{E(\Delta)}{4\pi^2},
\end{align*}
where $E(\Delta) = O(\Delta^2 |\ln\Delta|^2)$ as $\Delta \downarrow 0$ 
which, by using the definition of $b_1$ in (\ref{b1def}), can be written as 
\begin{align}\label{1overUDelta}
	\frac{1}{U} =  \frac{(\ln(\frac{\Delta}{32}))^2-b_1+E(\Delta)}{4\pi^2}. 
\end{align}
It is clear from (\ref{mAFeq}) that $\Delta(U) \downarrow 0$ as $U \downarrow 0$. 
Using this fact to pick the appropriate root when solving (\ref{1overUDelta}) for $\Delta$, we obtain
\begin{align}\label{Delta32eE}
		\Delta = 32 e^{-\sqrt{\frac{4 \pi^2}{U} + b_1 - E(\Delta)}}.
\end{align}

The relation (\ref{Delta32eE}) implies that $\Delta = O(e^{-\frac{2\pi}{\sqrt{U}}})$ and hence that $E(\Delta) = O(U^{-1} e^{-\frac{4\pi}{\sqrt{U}}})$ as $U \downarrow 0$. 
Substituting this back into (\ref{Delta32eE}), we find
$$\Delta(U) = 32 \exp\Bigg(-\sqrt{\frac{4 \pi^2}{U} + b_1 + O\bigg(\frac{e^{-\frac{4\pi}{\sqrt{U}}}}{U}\bigg)}\Bigg)
\qquad \text{as $U \downarrow 0$},$$
from which Proposition \ref{prop1} follows.

\section{Computing $a_0$: proof of Proposition \ref{a0prop}}\label{proofsec2}
In this section, we show that the constant $a_0$ defined in (\ref{a0def}) is given by 
\begin{equation}\label{a0conjecture}  
	a_0 = \frac{(\ln 2)^2}{\pi^2}-\frac1{24}. 
\end{equation} 
Recalling the definition (\ref{b1def}) of $b_1$, we see that Proposition \ref{a0prop} is a direct consequence of (\ref{a0conjecture}).

By the dominated convergence theorem and \eqref{a0def}, we have 
\begin{align*}
a_0 & = \lim_{s \downarrow 0} \int_{0}^{4} \bigg(N_0(\epsilon) - \frac{\ln(\frac{16}{\epsilon})}{2 \pi^2}\bigg)
\frac{1}{\epsilon^{1-s}}  d\epsilon = \lim_{s \downarrow 0} [J_1(s) - J_2(s)],
\end{align*}
where the functions $J_1$  and $J_2$  are defined by
\begin{align}\label{J1J2def} 
J_1(s) := \int_{0}^{4}  \frac{N_0(\epsilon)}{\epsilon^{1-s}}  d\epsilon,
\qquad J_2(s) := \int_{0}^{4} \frac{ \ln(\frac{16}{\epsilon})}{2\pi^2} \frac{1}{\epsilon^{1-s}}  d\epsilon\quad (s>0). 
\end{align}
The integral defining $J_2(s)$ is elementary, 
\begin{align} 
J_2(s) =\frac{4^{s}(1+2s\ln2)}{2\pi^2 s^2}   \quad (s>0), 
\end{align} 
and the first few terms in its Laurent series around $s = 0$ are
\begin{align*} 
J_2(s) = \frac{1}{2\pi^2s^2} +  \frac{2\ln 2}{\pi^2 s} + \frac{3(\ln 2)^2}{\pi^2} + \cdots.
\end{align*} 

The other integral can also be computed exactly, for arbitrary $s>0$; since this exact integral has other applications in condensed matter physics \cite{CLL2025}, we present it as a lemma. 

\begin{lemma} \label{lemma:key} The function $N_0(\epsilon)$ defined in \eqref{N0} satisfies
\begin{align}\label{J1} 
\int_{0}^{4}  \epsilon^{s-1} N_0(\epsilon) d\epsilon  = \frac{4^s}{8\pi}\left( \frac{\Gamma(\tfrac{s}{2})}{\Gamma(\tfrac12+\tfrac{s}{2})}\right)^2 \quad (s>0)
\end{align} 
with $\Gamma(z)$ the Euler Gamma function. 
\end{lemma} 

(The proof can be found at the end of this section.) 

Using Lemma \ref{lemma:key} and the Laurent series $\Gamma(\tfrac{s}{2})/\Gamma(\tfrac12+\tfrac{s}{2})=\frac2{\sqrt{\pi}s}+\frac{2\ln 2}{\sqrt{\pi}}+O(s)$,  we find 
$$
J_1(s) =  \frac{4^s}{8\pi}\left( \frac{\Gamma(\tfrac{s}{2})}{\Gamma(\tfrac12+\tfrac{s}{2})}\right)^2  = \frac{1}{2\pi^2s^2} +  \frac{2\ln 2}{\pi^2 s}  + \frac{4(\ln 2)^2}{\pi^2} -\frac1{24}+O(s)\quad \text{as $s \to 0$}.
$$
Thus, the divergent terms in $J_1(s)$ and $J_2(s)$ are the same, and the limit $s\downarrow 0$ of 
$$J_1(s)-J_2(s) = \frac{(\ln 2)^2}{\pi^2} -\frac1{24}+O(s) $$ 
yields \eqref{a0conjecture}, which in view of (\ref{b1def}) implies Proposition \ref{a0prop}.

\begin{proof}[Proof of Lemma~\ref{lemma:key}]
Let $s > 0$. Since $N_0(\epsilon)$ is even and has support in $[-4,4]$, it follows from (\ref{gN0}) that
\begin{align*}
\int_{-4}^4 g(\epsilon) N_0(\epsilon)d\epsilon = \int_{-\pi}^\pi \int_{-\pi}^\pi g(2 (\cos{k_1} + \cos{k_2})) \frac{dk_1dk_2}{(2\pi)^2}
\end{align*}
for any measurable function $g:\R \to \R$ such that $gN_0 \in L^1(\R)$.
Using this for $g(\epsilon)= \epsilon^{s-1} \theta(\epsilon)$, where $\theta$ denotes the Heaviside function, 
we obtain
$$
J_1(s)=\int_{-4}^4 \epsilon^{s-1} \theta(\epsilon) N_0(\epsilon)d\epsilon =  \int_{[-\pi,\pi]^2} \big(2(\cos{k_1} +\cos{k_2})\big)^{s-1} \theta\big(2(\cos{k_1} +\cos{k_2})\big)\frac{dk_1dk_2}{(2\pi)^2}
$$
for $s>0$. 
Since the function $2(\cos{k_1} +\cos{k_2})$ of $(k_1,k_2)\in[-\pi,\pi]^2$ is $\geq 0$ in the region $|k_1+k_2|\leq \pi$, $|k_1-k_2|\leq \pi$ and $<0$ otherwise, 
$$
J_1(s) =  \int_{\substack{|k_1+k_1|\leq \pi\\|k_1-k_1|\leq \pi}} \big(2(\cos{k_1} +\cos{k_2})\big)^{s-1} \frac{dk_1dk_2}{(2\pi)^2}. 
$$
Moreover, since  
$$
2(\cos{k_1} +\cos{k_2}) = 4\cos\left(\frac{k_1+k_2}{2}\right)\cos\left(\frac{k_1-k_2}{2}\right), 
$$
we can compute the latter integral by a change of variables $(k_1,k_2)\to (k_-,k_+)$ where $k_\pm = (k_1\pm k_2)/2$, which gives (the factor $2$ is the Jacobian determinant for this variable change),  
\begin{multline*} 
J_1(s) = 2\int_{[-\pi/2,\pi/2]^2}\left(4\cos\left(k_+\right)\cos\left(k_-\right)\right)^{s-1} \frac{dk_+dk_-}{(2\pi)^2}\\
=  \frac{4^{s}}{8\pi^2}\left( \int_{-\pi/2}^{\pi/2} \cos(k)^{s-1}dk\right)^2 =  \frac{4^{s}}{8\pi^2}\left( 2\int_{0}^{\pi/2} \cos(k)^{s-1} dk\right)^2 = \frac{4^{s}}{8\pi^2}\left( \frac{\Gamma(\tfrac12) \Gamma(\tfrac{s}{2})}{\Gamma(\tfrac{1}{2}+\tfrac{s}{2})}\right)^2 , 
\end{multline*} 
using a well-known integral representation of the Beta function B$(\tfrac12,\tfrac{s}{2})$ and its relation to the Gamma function. By inserting $\Gamma(\tfrac12)=\sqrt{\pi}$, we obtain  the result. 
\end{proof}

\bigskip\noindent

\noindent {\bf Acknowledgements.} 
{\it We acknowledge support from the Swedish Research Council (VR), Grants No. 2023-04726 (EL) and No. 2021-03877 (JL).}

\medskip\noindent
{\bf Conflicts of interest.} {\it  On behalf of all authors, the corresponding author states that there is no conflict of interest.}

\medskip\noindent
{\bf Data availability.} {\it Data sharing not applicable to this article as no datasets were generated or analyzed during the
current study.}

\appendix
\section{Dilogarithm}\label{polylogapp}
The dilogarithm $\Li_2(z)$ is defined by 
\begin{align}\label{polylogdef}
\Li_2(z) = \sum_{k=1}^\infty \frac{z^k}{k^2} \qquad \text{for $z\in\C$ such that $|z| < 1$}
\end{align}
and extended to $z \in \C \setminus [1, +\infty)$ by analytic continuation.
It obeys (see e.g. \cite{Z2007})
\begin{align}\label{polylogderivatives}
\frac{d}{dz}\Li_2(z) = -\frac{\ln(1-z)}{z} 
\end{align}
and
\begin{align}\label{dilogidentity}
\Li_2\Big(-\frac{1}{z}\Big) = -\Li_2(-z) - \frac{1}{2}(\ln z)^2 - \frac{\pi^2}{6} \qquad \text{for $z \in \C \setminus (-\infty,0]$}.
\end{align}
The identity (\ref{dilogidentity}) together with (\ref{polylogdef}) implies the asymptotic formula
\begin{align}\label{Li2asymptotics}
\Li_2(-z) = -\frac{1}{2}(\ln z)^2 - \frac{\pi^2}{6} + O(z^{-1}) \qquad \text{as $z \to \infty$}
\end{align}
uniformly for $\arg(z) \in (-\pi, \pi)$. 

\bibliographystyle{amsplain}

\end{document}